\documentclass[conference]{IEEEtran}
\IEEEoverridecommandlockouts

\usepackage{cite}
\usepackage{amsmath,amssymb,amsfonts}
\usepackage{algorithmic}
\usepackage{graphicx}
\usepackage{textcomp}

\usepackage[utf8]{inputenc} 
\usepackage[T1]{fontenc}    

\usepackage[linesnumbered,ruled,vlined ]{algorithm2e}

\SetCommentSty{mycommfont}
\usepackage{subcaption}
\usepackage{tcolorbox}

\usepackage{hyperref}

\newcommand{\K}{\mathbf{K}} 
\newcommand{\A}{\mathbf{A}} 
\newcommand{\B}{\mathbf{B}} 
\newcommand{\E}{\mathbb{E}}

\newcommand{\x}{\mathbf{x}} 
 
\newcommand{\M}{\mathbf{M}} 
 
\newcommand{\y}{\mathbf{y}}

\newcommand{\C}{\mathbf{C}}
\newcommand{\rr}{\mathbf{r}}
\newcommand{\z}{\mathbf{z}}
\newcommand{\p}{\mathbf{p}} 
\newcommand{\I}{\mathbf{I}}
\newcommand{\R}{\mathbb{R}}
\newcommand{\h}{\mathbf{h}}

\DeclareMathOperator*{\argmin}{arg\,min}

\title{Serverless Straggler Mitigation using Local Error-Correcting Codes}
\author{\IEEEauthorblockN{Vipul Gupta$^\star$, Dominic Carrano$^\star$, Yaoqing Yang, Vaishaal Shankar, Thomas Courtade and Kannan Ramchandran
\thanks{$^\star$Equal contribution.}
}
\IEEEauthorblockA{Department of EECS, UC Berkeley
}
}

\usepackage{amsthm}
\newtheorem{theorem}{Theorem}
\newtheorem{corollary}{Corollary}

\newtheorem{definition}{Definition}
\newtheorem{remark}{Remark}

\usepackage{xcolor}

\begin{document}

\maketitle

\begin{abstract}

Inexpensive cloud services, such as serverless computing, are often vulnerable to straggling nodes that increase end-to-end latency for distributed computation.
We propose and implement simple yet principled approaches for straggler mitigation in serverless systems for matrix multiplication and evaluate them on several common applications from machine learning and high-performance computing. The proposed schemes are inspired by error-correcting codes and employ parallel encoding and decoding over the data stored in the cloud using serverless workers. This creates a fully distributed computing framework without using a master node to conduct encoding or decoding, which removes the computation, communication and storage bottleneck at the master.
On the theory side, we establish that our proposed scheme is asymptotically optimal in terms of decoding time and provide a lower bound on the number of stragglers it can tolerate with high probability. 
Through extensive experiments, we show that our scheme outperforms existing schemes such as speculative execution and other coding theoretic methods by at least $25\%$.
\end{abstract}

\section{Introduction}


We focus on a recently introduced cloud service called \emph{serverless computing} for general distributed computation. Serverless systems have garnered significant attention from industry (e.g., Amazon Web Services (AWS) Lambda, Microsoft Azure Functions, Google Cloud Functions) as well as the research community (see, e.g.,~\cite{pywren,serverless_computing,numpywren, hellerstein2018serverless,oversketch,osn, berkeley_view,serverless_faaster_better_cheaper}).
Serverless platforms\footnote{The name serverless is an oxymoron since all the computing is still done on servers, but the name stuck as it abstracts away the need to provision or manage servers.} penetrate a large user base by removing the need for complicated cluster management while providing greater scalability and elasticity \cite{serverless_computing,pywren,numpywren}. 
For these reasons, serverless systems are expected to abstract away today's cloud servers in the coming decade just as cloud servers abstracted away physical servers in the past decade \cite{serverless_faaster_better_cheaper,berkeley_view,medium_future_serverless}. 

However, system noise in inexpensive cloud-based systems results in subsets of slower nodes, often called \emph{stragglers}, which significantly slow the computation.
This system noise is a result of limited availability of shared resources, network latency, hardware failure, etc. \cite{tailatscale, hoefler}. 
Empirical statistics for worker job times are shown in Fig. \ref{fig:stragglers} for AWS Lambda. Notably, there are a few workers ($\sim$$2\%$) that take much longer than the median job time, severely degrading the overall efficiency of the system. 

Techniques like speculative execution have been traditionally used to deal with stragglers (e.g., in Hadoop MapReduce \cite{mapreduce} and Apache Spark \cite{spark}). Speculative execution works by detecting workers that are running slowly, or will slow down in the future, and then assigning their jobs to new workers without shutting down the original job. The worker that finishes first submits its results. This has several drawbacks:  constant monitoring of jobs is required, which is costly when the number of workers is large. 
Monitoring is especially difficult in serverless systems where worker management is done by the cloud provider and the user has no direct supervision over the workers.
Moreover, it is often the case that a worker straggles only at the end of the job (say, while communicating the results).  By the time the job is resubmitted, the additional communication and computational overhead would have decreased the overall efficiency of the system. \looseness=-1


\begin{figure}
        \centering
        \includegraphics[scale=0.43]{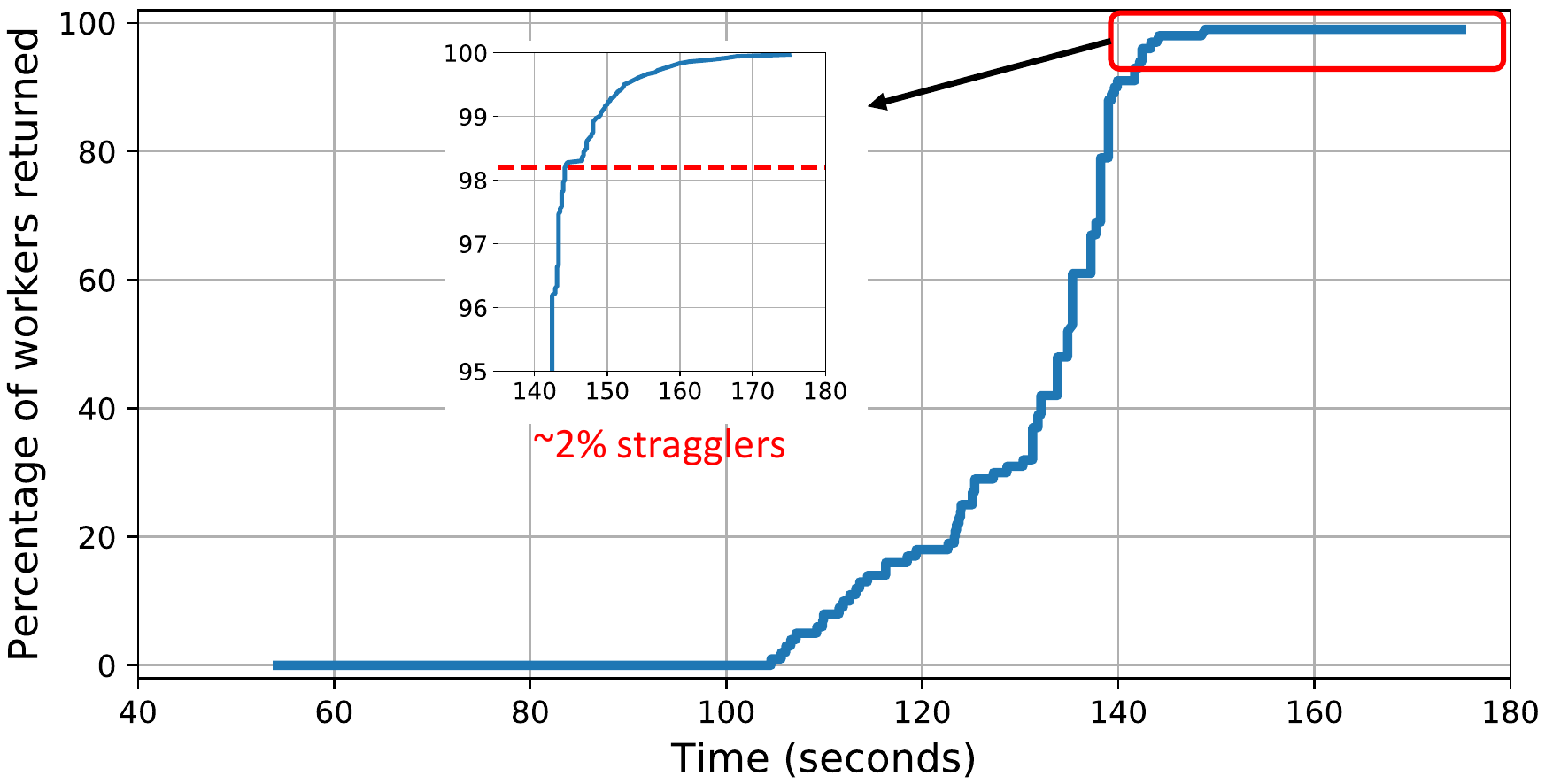}
        \caption{\small Distribution of job completion times for distributed matrix multiplication over 3600 AWS Lambda workers averaged over 10 trials. The median job time is $\sim$$135$ seconds, while around $2\%$ of the nodes straggle consistently.
        }
        \label{fig:stragglers}
\end{figure}

\subsection{Existing Work}

\begin{figure*}[t]
    \centering
    \includegraphics[height=1.4in]{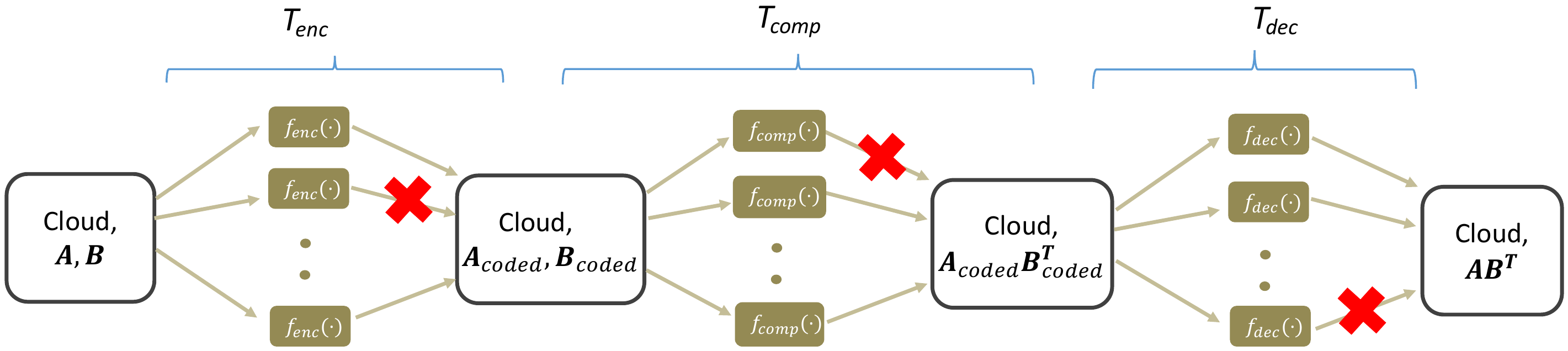}
    \caption{\small Typical workflow on a serverless system for computing the matrix multiplication $\A\B^T$. Here, $f_{\text enc}, f_{\text comp}$ and $f_{\text dec}$ denote the functions corresponding to encoding, computation, and decoding, respectively, that are employed at the serverless workers (in parallel on different data points).
    Whereas most existing schemes focus on minimizing time required to compute the product ($T_{comp}$), our focus is on minimizing the end-to-end latency that involves parallel encoding ($T_{enc}$) and decoding ($T_{dec}$) times as well. 
    }
    \label{fig:local_codes}
\end{figure*}

Error correcting codes are a linchpin of digital transmission and storage technologies, vastly improving their efficiency compared to uncoded systems. Recently, there has been a significant amount research focused on applying coding-theoretic ideas to introduce redundancy into distributed computation for improved straggler and fault resilience, see, e.g., \cite{kangwook1,grad_coding,kangwook2,tavor,poly_codes,matdot, bartan2019polar, jeong2018locally,coded2.5d,krishna_poly_codes,jingge, dutta_shortdot, rashmi_inference, yang_substitute_dec, grover_inverse, abbe_comp_comm}.

This line of work focuses on cloud computing models consistent with first-generation cloud platforms (i.e., ``serverful" platforms), where the user is responsible for node management through a centralized master node that coordinates encoding, decoding and any update phases.
Accordingly, most existing  schemes typically employ variants of Maximum Distance Separable (MDS) codes, and have focused on optimizing the recovery threshold (i.e., minimum number of machines needed to do a task) of the algorithm, e.g. \cite{poly_codes,matdot}.
This is equivalent to minimizing the compute time while assuming that the encoding/decoding times are negligible.
When the system size is relatively small, the encoding/decoding costs can be safely ignored. 
However, the encoding/decoding costs of such coded computation schemes scale with the size of the system, and hence this assumption does not hold anymore for serverless systems that can invoke tens of thousands of workers \cite{numpywren,berkeley_view,osn}. Furthermore, existing schemes require a powerful master with high bandwidth and large memory to communicate and store all the data to perform encoding and decoding locally. This goes against the very idea of massive scale distributed computation.
Therefore, coding schemes designed for serverful systems cannot guarantee low end-to-end latency in terms of total execution time for large-scale computation in serverless systems.
\looseness=-1

To formalize this problem, we consider the typical workflow of a serverless system for the task of matrix-matrix multiplication (see Fig.~\ref{fig:local_codes}).
First, worker machines read the the input data from the cloud, jointly encode the data, and write the encoded data to the cloud ($T_\text{enc}$). 
Then, the workers start working on their tasks using the encoded data, and write back the product of coded matrices back to the cloud memory.
Denote the joint compute time (including the time to communicate the task results to the cloud) $T_{\text comp}$.
Once a \emph{decodable} set of task results are collected, the workers start running the decoding algorithm to obtain the final output (which takes $T_\text{dec}$ time). 
Note that all of these phases are susceptible to straggling workers.
Hence, one can write the total execution time of a coded computing algorithm as
$T_\text{tot,coded} = T_\text{enc} + T_{\text comp} + T_\text{dec}.$
The key question that we ask is how to minimize end-to-end latency, $T_\text{tot, coded}$, that comprises encoding, decoding and computation times, where all of these phases are performed in parallel by serverless workers.

\subsection{Main Contribution}

In this work, we advocate principled, coding-based approaches to accelerate distributed computation in serverless computing. Our goals span both theory and practice: we develop coding-based techniques to solve common machine learning problems on serverless platforms in a fault/straggler resilient manner, analyze their runtime and straggler tolerance, and implement them on AWS Lambda for several popular applications. 

Generally, computations underlying several linear algebra and optimization problems tend to be iterative in nature. With this in mind, we aim to develop general coding-based approaches for straggler-resilient computation which meet the following criteria: (1) Encoding over big datasets should be performed once.  In particular,  the cost for encoding the data for straggler-resilient computation will be amortized over iterations. (2) Encoding and decoding should be low-complexity and require at most linear time and space in the size of the data. (3) Encoding and decoding should be amenable to a parallel implementation. This final point is particularly important when working with large datasets on serverless systems due to the massive scale of worker nodes and high communication latency. \looseness=-1



It is unlikely that there is a ``one-size-fits-all" methodology which meets the above criteria and introduces straggler resilience for any problem of interest.  Hence, we propose to focus our efforts on a few fundamental operations including matrix-matrix multiplication and matrix-vector multiplication, since these form atomic operations for many large-scale computing tasks. Our developed algorithms outperform speculative execution and other popular coding-based straggler mitigation schemes by at least $25\%$. 
We demonstrate the advantages of using the developed coding techniques on several applications such as alternating least squares, SVD, Kernel Ridge Regression, power iteration, etc.

\section{Straggler Resilience in Serverless Computing Using Codes}
\subsection{Distributed Matrix-Vector Multiplication}

\begin{figure*}[h!]
    \centering
    \begin{subfigure}[t]{0.45\textwidth}
        \centering
        \includegraphics[scale=0.45]{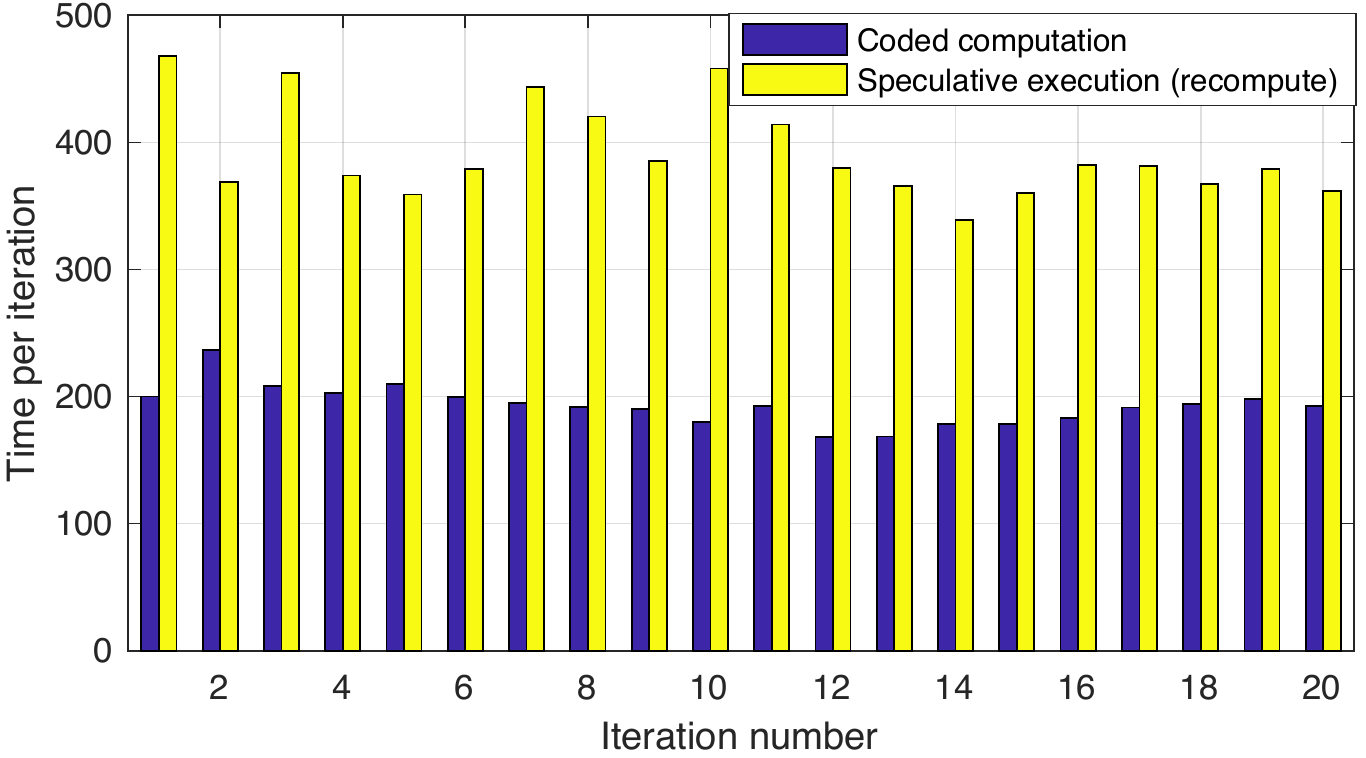}
        \caption{Per iteration time for power iteration}
    \end{subfigure}
    ~~
    \begin{subfigure}[t]{0.45\textwidth}
        \centering
        \includegraphics[scale=0.45]{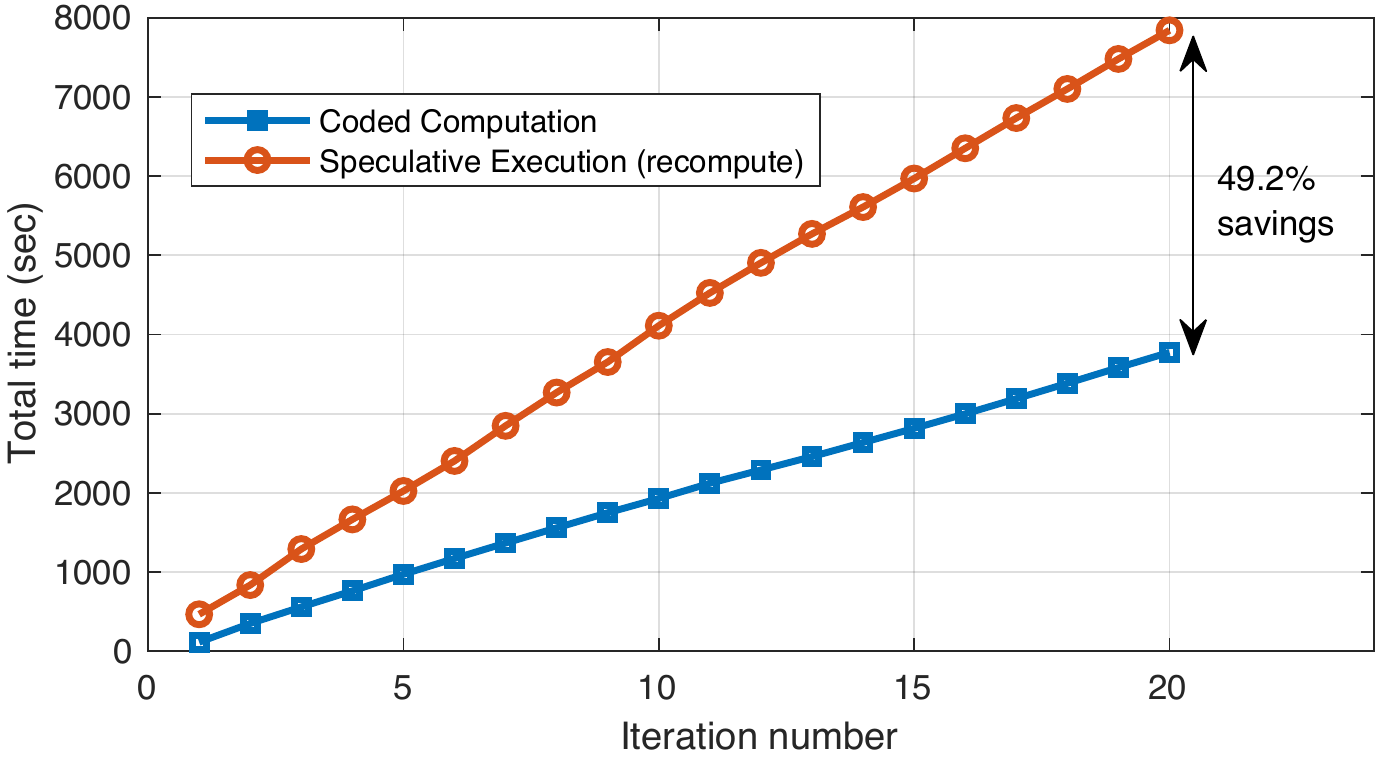}
        \caption{Running time for power iteration for 20 iterations}
    \end{subfigure}
    \caption{{\small Coded computing versus speculative execution for power iteration on a matrix of dimension 0.5 million for 20 iterations.}}
\label{fig:PI_stats}
\end{figure*} 

The main objective of this section is to show that coding schemes can hugely benefit serverless computing by implementing coded matrix-vector multiplication on AWS Lambda.  
Computing $\y = \A\x$, for a large matrix $\A$, is a frequent bottleneck of several popular iterative algorithms such as gradient descent, conjugate gradient, power iteration, etc. Many coding theory based techniques for straggler-resilient matrix vector multiplication have been proposed in the literature (e.g. see \cite{kangwook1, grover_inverse,dutta_shortdot,tavor}). 
We refer the reader to Fig. 2 in \cite{kangwook1} for an illustration.
Fortunately, many of these schemes can be directly employed in serverless systems since the encoding can be done in parallel and the decoding over the resultant output for computing $\y$ is inexpensive as it is performed over a vector. Note that a direct applicability is not true for all operations (such as matrix-matrix multiplication), as we will see later in Section \ref{sec:local_matmul}.  \looseness=-1


To illustrate the advantages of coding techniques over speculative execution, we implement power iteration on the serverless platform AWS Lambda. 
Power iteration requires a matrix-vector multiplication in each iteration and gives the dominant eigenvector and corresponding eigenvalue of the matrix being considered.
Power iteration constitutes an important component for several popular algorithms such as PageRank and Principal Component Analysis (PCA). PageRank is used by Google to rank documents in their search engine \cite{page1999pagerank} and by Twitter to generate recommendations of who to follow \cite{gupta2013wtf}. 
PCA is commonly employed as a means of dimensionality reduction in applications like data visualization, data compression and noise reduction \cite{pca_svd_ml_app}.

We applied power iteration to a square matrix of dimension $(\text{0.5 million})^2$ using 500 workers on AWS Lambda in the Pywren framework \cite{pywren}. A comparison of compute times of coded computing with speculative execution is shown in Fig. \ref{fig:PI_stats}, where a $2\times$ speedup is achieved\footnote{For our experiments on matrix-vector multiplication, we used the coding scheme proposed in \cite{tavor} due to its simple encoding and decoding that takes linear time. However, we observed that using other coding schemes that are similar, such as the one proposed in \cite{kangwook1}, result in similar runtimes.}.
Apart from being significantly faster than speculative execution, another feature of coded computing is reliability, that is, almost all the iterations take a similar amount of time ($\sim$$200$ seconds) compared to speculative execution, the time for which varies between 340 and 470 seconds. We demonstrate this feature of coded computing throughout our experiments in this paper. 

\subsection{Distributed Matrix-Matrix Multiplication}\label{sec:local_matmul}

Large-scale matrix-matrix multiplication is a frequent computational bottleneck in several problems in machine learning and high-performance computing and has received significant attention from the coding theory community (e.g. see \cite{kangwook2,tavor,poly_codes,matdot,bartan2019polar,jeong2018locally,coded2.5d}). The problem is computing
\begin{equation}
\A \mathbf{B}^T = \C, ~\text{where}~ \A \in \R^{m \times n} ~\text{and}~ \B \in \R^{\ell \times n}.
\label{matmul_equation}
\end{equation} 

{\bf Proposed Coding Scheme}: For straggler-resilient matrix multiplication, we describe our easy-to-implement coding scheme below. First, we encode the row-blocks of $\A$ and $\B$ in parallel by inserting a parity block after every $L_A$ and $L_B$ blocks of $\A$ and $\B$, respectively, where $L_A$ and $L_B$ are parameters chosen to control the amount of redundancy the code introduces. This produces encoded matrices $\A_\text{coded}$ and $\B_\text{coded}$. As $L_A$ and $L_B$ are increased, the parity blocks become more spread out, and the code has less redundancy. For example, when $L_A = L_B = 1$, every row of the matrices $\A$ and $\B$ is duplicated (and, hence, has $100\%$ redundancy). At the other extreme, when $L_A$ and $L_B$ are set equal to the number of row-blocks in $\A$ and $\B$, respectively, there is only one parity row-block added in $\A$ and $\B$, and thus, the code exhibits minimum possible redundancy. 
In Fig. \ref{fig:mat_mul_coding}, an example of the encoded matrix $\A_\text{coded}$ and the resultant output matrix $\C_\text{coded}$ is shown for the case when $\A = \B$ and $L_A=2$.

\begin{figure*}[h!]
    \centering
    \includegraphics[scale=0.5]{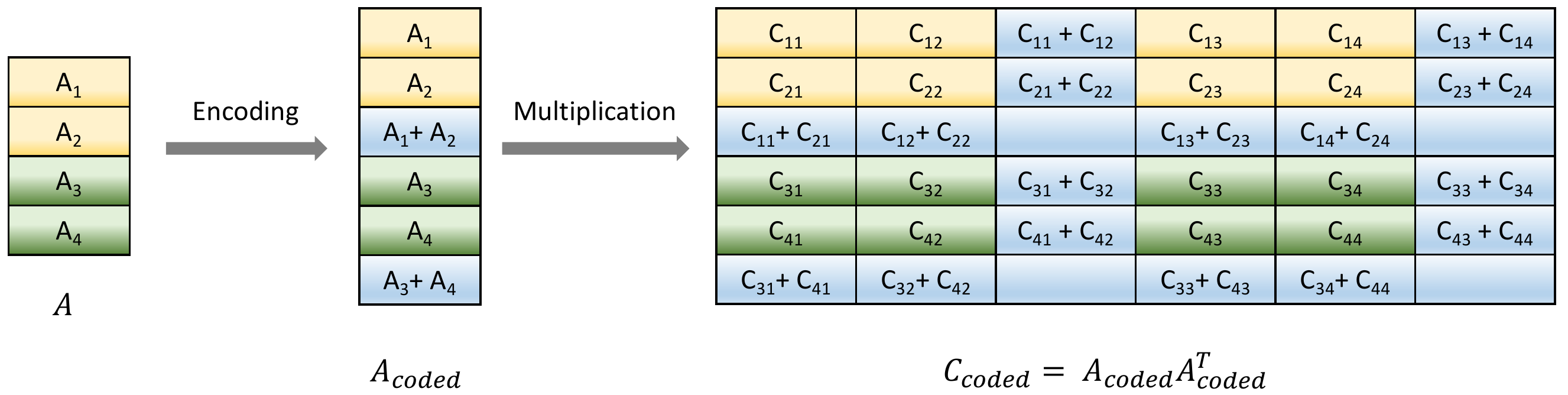}
    \caption{\small Computing $\C = \A\A^T$ where $\A$ is divided into four row-blocks and $L_A = 2$. Here, $\C_{ij} = \A_i\A_j^T$. Locally encoding the rows of $\A$ leads to a locally recoverable code in the output $\C_\text{coded}$.}
\label{fig:mat_mul_coding}
\end{figure*} 

Note the locally recoverable structure of $\C_\text{coded}$: to decode one straggler, only a subset of blocks of $\C_\text{coded}$ need to be read. In Fig. \ref{fig:mat_mul_coding}, for example, only two blocks need to be read to mitigate a straggler. This is unlike polynomial codes which are MDS in nature and, hence, are optimal in terms of recovery threshold but require reading all the blocks from the output matrix while decoding. 
The locally recoverable structure of the code makes it particularly amenable to a parallel decoding approach: $\C_\text{coded}$ consists of $(L_A + 1) \times (L_B + 1)$ submatrices, each of which can be separately decoded in parallel. In Fig. \ref{fig:mat_mul_coding}, there are four such submatrices. We use a simple peeling decoder (for example, see \cite{kangwook2,tavor}) to recover the systematic part of each $(L_A + 1) \times (L_B + 1)$ submatrix, constructing the final result matrix $\C$ from these systematic results.

In the event that any of the submatrices are not decodable due to a large number of stragglers, we recompute the straggling outputs. 
Thus, choosing $L_A$ and $L_B$ presents a trade-off. We would like to keep them small so that we can mitigate more stragglers without having to recompute, but smaller $L_A$ and $L_B$ imply more redundancy in computation and is potentially more expensive. For example, $L_A = L_B = 5$ implies $44\%$ redundancy. 
Later, we will show how to choose the parameters $L_A$ and $L_B$ given an upper bound on the probability of encountering a straggler in the serverless system. We will also prove that with the right parameters, the probability of not being able to decode the missing blocks is negligible. 


\begin{figure}[h!]
    \centering
    \includegraphics[scale=0.65]{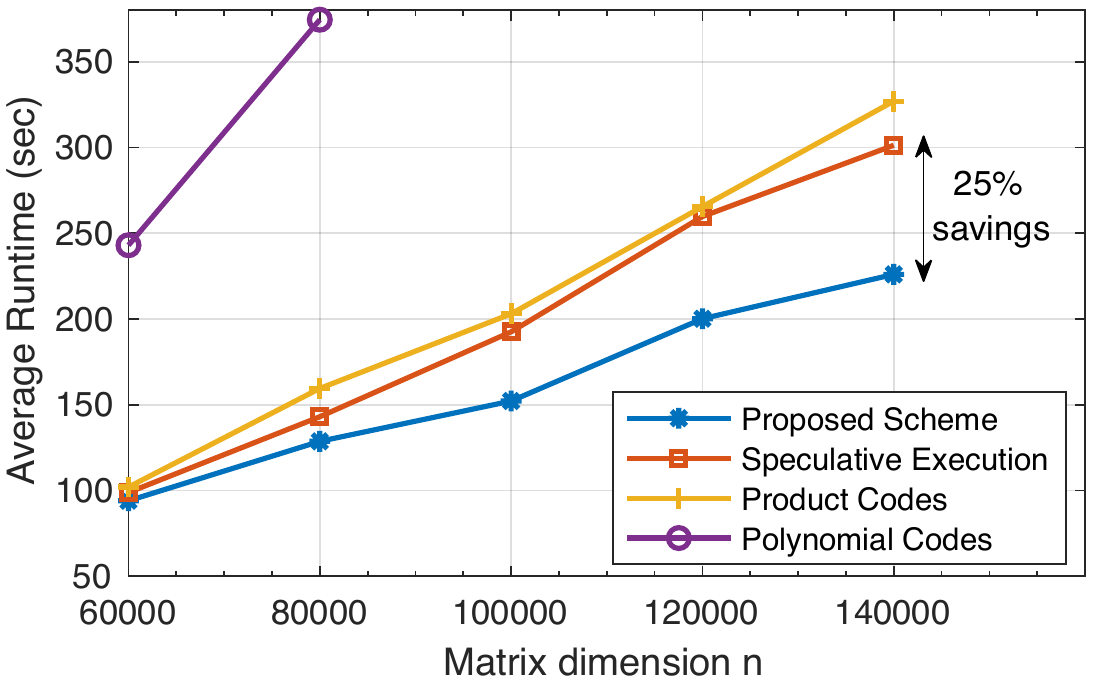}
    \caption{\small Comparison of average runtimes of proposed schemes versus existing schemes for multiplication of two square matrices. For large matrix dimensions, decoding with a polynomial code is not feasible since the master node cannot store all the data locally.}
\label{fig:mat_mul_stats}
\end{figure} 


We refer to the proposed coding scheme in Fig. \ref{fig:mat_mul_coding} as the \textit{local product code}. In Fig. \ref{fig:mat_mul_stats}, we compare the local product code with speculative execution, and existing popular techniques for coded matrix multiplication such as polynomial codes \cite{poly_codes} and product codes \cite{kangwook2}. 
In our experiment, we set $\A$ ($= \B$) to be a square matrix with $L_A = 10$, implying $21\%$ redundancy.
Product codes and polynomial codes were also designed such that the amount of redundancy was $\geq 21\%$. Accordingly, we wait for $79\%$ of the workers to return before starting to recompute in the speculative execution-based approach so that all the methods employed had the same amount of redundancy. We note that the coding-based approach performs significantly better than existing coding-based schemes and at least $25\%$ better than the speculative execution-based approach for large matrix dimensions\footnote{A working implementation of the proposed schemes is available at
https://github.com/vvipgupta/serverless-straggler-mitigation}.

Another important point to note is that existing coding-based approaches perform worse than speculative execution. This is because of the decoding overhead of such schemes. Product codes have to read the entire column (or row) block of $\C_\text{coded}$ and polynomial codes have to read the entire output $\C_\text{coded}$ to decode one straggler. In serverless systems, where workers write their output to a cloud storage and do not communicate directly with the master owing to their `stateless' nature, this results in a huge communication overhead. In fact, for polynomial codes, we are not even able to store the entire output in the memory of the master for larger values of $n$. For this reason, we do not have any global parities---that require reading all the blocks to decode the stragglers---in the proposed local product code. 
Note that existing coding schemes with locality, such as \cite{tavor} and \cite{jeong2018locally}, also have global parities which are dispensable in serverless and, thus, have high redundancy. This is because such schemes were designed for serverful systems where the decoding is not fully distributed. 
Moreover, we show in the next section that local product codes are asymptotically optimal in terms of locality for a fixed amount of redundancy.
In the event the output is not locally decodable in local product codes, we restart the jobs of straggling workers. 
However, we later show that such an event is unlikely if the parameters $L_A$ and $L_B$ are chosen properly. 

\begin{remark}
To mitigate stragglers during encoding and decoding phases, we employ speculative execution. However, in our experiments, we have observed that encoding and decoding times have negligible variance and do not generally suffer from stragglers. This is because the number of workers required during encoding and decoding phases is relatively small (less than $10\%$ of the computation phase) with smaller job times due to locality. The probability of encountering a straggler in such small-scale jobs is extremely low. 
\end{remark}

\begin{remark}
It has been well established in the literature that blocked partitioning of matrices is communication efficient for distributed matrix-matrix multiplication both in the serverful \cite{2.5d,summa} and serverless \cite{oversketch} settings. 
Even though in Fig. \ref{fig:mat_mul_coding} we show partitioning of $\A$ into row-blocks for clarity of exposition, we further partition the input matrices $\A$ (and $\B$) into square blocks in all our experiments and perform block-wise distributed multiplication. 
\end{remark}

\section{Theoretical Analysis of Local Product Codes}

\subsection{Optimality of Local Product Codes}
In coding-theoretic terminology, a locally recoverable code (LRC) is a code where each symbol is a function of small number of other symbols. This number is referred to as the locality, $r$, of the code. In the context of straggler mitigation, this means that each block in $\C_\text{coded}$ is a function of only a few other blocks. Hence, to decode one straggler, one needs to read only $r$ blocks. In the example of Fig. \ref{fig:mat_mul_coding}, the locality is $r=2$ since each block of $\C_\text{coded}$ can be recovered from two other blocks. In general, the locality of the local product code is $\min(L_A, L_B)$. Another important parameter of a code is its minimum distance, $d$, which relates directly to the number of stragglers that can be recovered in the worst case. Specifically, to recover the data of $e$ stragglers in the worst case, the minimum distance must satisfy $d\geq e+1$.

For a fixed redundancy, Maximum Distance Separable (MDS) codes attain the largest possible minimum distance $d$, and thus, are able to tolerate the most stragglers in the worst case. Many straggler mitigation schemes are focused on MDS codes and have gained significant attention, such as polynomial codes \cite{poly_codes}. However, such schemes are not practical in the serverless case since they ignore the encoding and decoding costs. Moreover, as seen from Fig. \ref{fig:mat_mul_stats}, it is better to restart the straggling jobs than to use the parities from polynomial or product codes since the communication overhead during decoding is high.

Hence, in serverless systems, the locality $r$ of the code is of greater importance since it determines the time required to decode a straggler. For any LRC code, the following relation between $d$ and $r$ is satisfied \cite{papailiopoulos2014locally,gopalan2012locality}
\begin{equation}
d \leq n - k - \bigg\lceil \frac{k}{r} \bigg\rceil + 2,
\label{eq:singleton}
\end{equation}
where $k$ is the number of systematic data blocks and $n$ is the total number of data blocks including parities. Now, since we want to tolerate at least one straggler, the minimum distance must satisfy $d \geq 2$. Using $\lceil k/r \rceil \geq k/r$, we conclude that
$n - k - \frac{k}{r}\geq 0$
or, equivalently, 
\begin{equation}
r \geq \frac{k}{n-k}.
\label{r_lower_bound}
\end{equation}
Now, in the case of the local product code, each of the submatrices that can be decoded in parallel represent a product code with $k = L_AL_B$ and $n = (L_A+1)(L_B + 1)$. In Fig. \ref{fig:mat_mul_coding}, there are four locally decodable submatrices with $L_A = L_B = 2, k = 4$ and $n = 9$. Also, we know that the locality for each of the submatrices is $\min(L_A, L_B)$ and hence this is the locality for the local product code. 

Next, we want to compare the locality of the local product code with any other coding scheme with the same parameters, that is, $k = L_AL_B$ and $n = (L_A+1)(L_B + 1)$. Using Eq. \ref{r_lower_bound}, we get
\begin{align*}
r &\geq \frac{L_AL_B}{(L_A+1)(L_B+1) - L_AL_B} = \frac{L_AL_B}{L_A + L_B + 1} \\
& \geq \frac{\min(L_A,L_B)}{2 + o(1)}.
\end{align*}
Thus the locality of local product codes is optimal (within a constant factor) since it achieves the lower bound of locality $r$ for all LRC codes. This is asymptotically better than, say, a local version of polynomial codes (that is, each submatrix of $\C_\text{coded}$ is a polynomial code instead of a product code) for which the locality is $L_AL_B$ since it needs to read all $L_AL_B$ blocks to mitigate one straggler \cite{poly_codes}.

Having shown that local product codes are asymptotically optimal in terms of decoding time, we further quantify the decoding time in the serverless case through probabilistic analysis next. 

\subsection{Decoding Costs}\label{sec:decoding_costs}

Stragglers arise due to system noise which is beyond the control of the user (and maybe even the cloud provider, for example, unexpected network latency or congestion due to a large number of users). However, a good estimate for an upper bound on the number of stragglers can be obtained through multiple experiments. In our theoretical analysis, we assume that the probability of a given worker straggling is fixed as $p$, and that this happens independently of other workers. In AWS Lambda, for example, we obtain an upper bound on the number of stragglers through multiple trial runs and observe that less than $2\%$ of the nodes straggle in most trials (also noted from Fig. \ref{fig:stragglers}). Thus, a conservative estimate of $p=0.02$ is assumed for AWS Lambda. 

Given the high communication latency in serverless systems, codes with low I/O overhead are highly desirable, making locally recoverable codes a natural fit.
For local product codes, say the decoding worker operates on a grid of $n= (L_A + 1) \times (L_B + 1)$ blocks.  
If a decoding worker sees a single straggler, it reads $\min(L_A,L_B)$ blocks to recover it. However, when there are more than one stragglers, at most $L = \max(L_A, L_B)$ block reads will occur per straggler during recovery. For example, if $L_A>L_B$ and there are two stragglers in the same row, the decoding worker read $L_A$ rows per straggler. 
Thus, if a decoding worker gets $S$ stragglers, a total of at most $SL$ block reads will occur---there are at most $L$ block reads for each of the $S$ stragglers. Since the number of stragglers, $S$, is random, the number of blocks read, say $R$, is also random. Note that $R$ scales linearly with the communication costs.

In Theorem \ref{thm:decoding_costs}, we quantify the decoding costs for local product codes; specifically, we show that the probability of a decoding worker reading a large number of blocks is small.
\begin{theorem}\label{thm:decoding_costs}
Let $p$ be the probability that a serverless worker straggles independently of others, and $R$ be the number of blocks read by a decoding worker working on $n = (L_A+1)(L_B+1)$ blocks. Also, let $L = \max(L_A,L_B)$. Then, the probability that the decoding worker has to read more than $x$ blocks is upper bounded by
\begin{align*}
    \Pr(R \geq x) \leq \left(\dfrac{x}{npL}\right)^{-x / L}e^{-\frac{x}{L} + np}
\end{align*}
\end{theorem}
\begin{proof}
See Section \ref{proof:decoding_costs}.
\end{proof}
\noindent Theorem \ref{thm:decoding_costs} provides a useful insight about the performance of local product codes: the probability of reading more than $x$ blocks during decoding decays decays to zero at a super-exponential rate.
Note that for the special (and more practical) case of $L_A = L_B = L$, the number of blocks read per straggler is exactly $L$ and thus $\E[R] = \E[SL] = npL$. Thus, using Theorem \ref{thm:decoding_costs}, we can obtain the following corollary.
\begin{corollary} 
For any $\epsilon>0$ and $L = L_A = L_B$, the probability that the decoding worker reads $\epsilon L$ more blocks than the expected $\E[R]$ blocks is upper bounded by
\begin{equation*}
\Pr(R \geq \E[R] + \epsilon L) \leq \left(1 + \dfrac{\epsilon}{np}\right)^{-np-\epsilon}e^{-\epsilon}.
\end{equation*}
For $\epsilon=np$, this becomes
\begin{equation*}
\Pr(R \geq 2\E[R]) \leq \frac{1}{(4e)^{np}}.
\end{equation*}
\end{corollary}

In Fig. \ref{fig:decode_costs}, we plot the upper bound on $\Pr(R \geq x)$ for different values of $x$. The values of $n$ and $L$ were chosen to be consistent with the experiments in Fig. \ref{fig:mat_mul_stats}, where $L_A = L_B = 10$, so that the maximum number of blocks read per straggler is $L = 10$ and the number of blocks of $\C_\text{coded}$ per decoding worker is $n =121$. Additionally, we used $p = .02$ as obtained through extensive experiments on AWS Lambda (see Fig. \ref{fig:stragglers}). In a polynomial code with the same locality, $100$ blocks would be read to mitigate any straggler by a decoding worker. For the local product code, the probability that $100$ blocks are read is upper bounded by $\Pr(R \geq 100) \leq 3.5 \times 10^{-10}$.

\begin{figure}[t]
    \centering
    \includegraphics[width=0.8\linewidth]{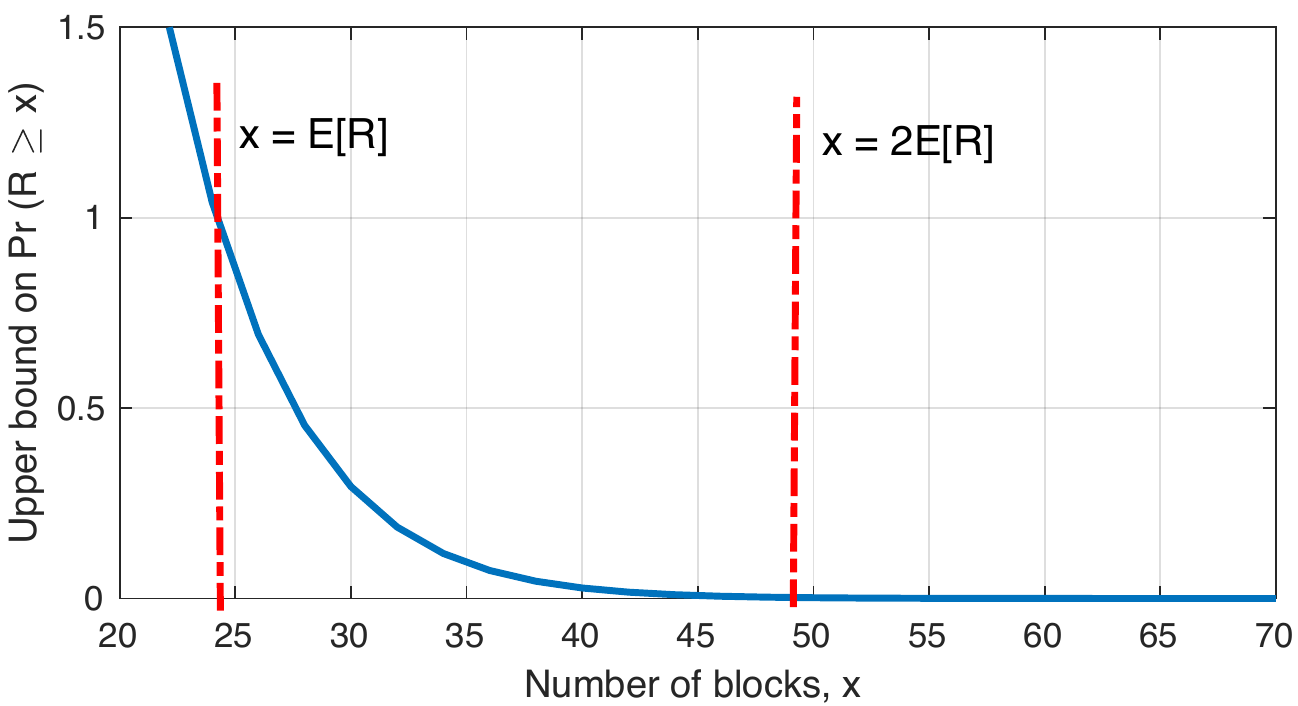}
    \caption{\small Probabilistic upper bound on the number of blocks read, $R$, by a decoding worker from Theorem \ref{thm:decoding_costs} shown for $L = 10$, $n=121$, and $p=0.02$. Here, $\Pr(R\geq 2\E[R]) \leq 3.1\times 10^{-3}.$}
\label{fig:decode_costs}
\end{figure} 

\subsection{Straggler Resiliency of Local Product Codes}\label{sec:straggler_resiliency}

\begin{figure}[t]
    \centering
    \includegraphics[width=0.9\linewidth]{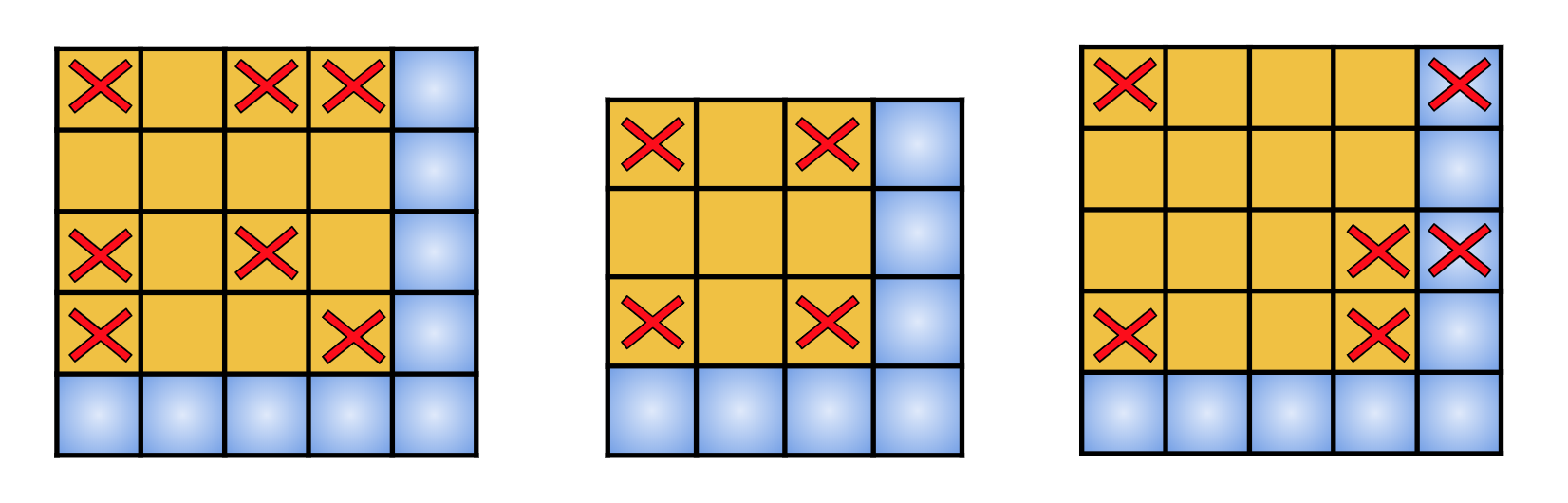}
    \caption{\small Some examples of undecodable sets, as viewed from a single decoding worker's $(L_A + 1)\times (L_B + 1)$ grid. The yellow blocks correspond to the systematic part of the code, and blue blocks to the parity. Blocks marked with an "X" are stragglers.}
\label{fig:trap_sets}
\end{figure} 

To characterize the straggler resiliency of local product codes, we turn our focus to finding the probability of encountering an \textit{undecodable set}: a configuration of stragglers that cannot be decoded until more results arrive. 

\begin{definition}
\textbf{Undecodable set}: Consider a single decoding worker that is working on $n$ blocks, arranged in an $(L_A + 1) \times (L_B + 1)$ grid, and let $S$ be the number of missing workers. The decoding worker's blocks are said to form an $S$-undecodable set if we need to wait for more workers to arrive to decode all the $S$ missing blocks.
\label{def:undec_set}
\end{definition}



Some examples of undecodable sets are shown in Fig. \ref{fig:trap_sets}.
In an $S$-undecodable set, it is possible that some of the $S$ stragglers are decodable, but there will always be some stragglers that are preventing each other from being decoded. For the local product code, an individual straggler is undecodable if and only if there is at least one other straggler in both its row and column, because the code provides a single redundant block along each axis that can be used for recovery. This implies that a decoding worker must encounter at least three stragglers for one of them to be undecodable. However, the code can always recover any three stragglers through the use of a peeling decoder \cite{kangwook2,tavor}. While the three stragglers may share a column or row and be in an "interlocking" configuration, such as those shown in Fig. \ref{fig:three_stragglers}, two of the three can always be recovered, or "peeled off". Using these blocks, the straggler that was originally undecodable can be recovered. This provides a key result: all undecodable sets consist of four or more stragglers. Equivalently, given $S \leq 3$, the probability of being unable to decode is zero. This can also be noted directly from the fact the the minimum distance of a product code with one parity row and column is four, and hence, it can tolerate any three stragglers \cite{kangwook2}.

\begin{figure}[t]
    \centering
    \includegraphics[width=0.95\linewidth]{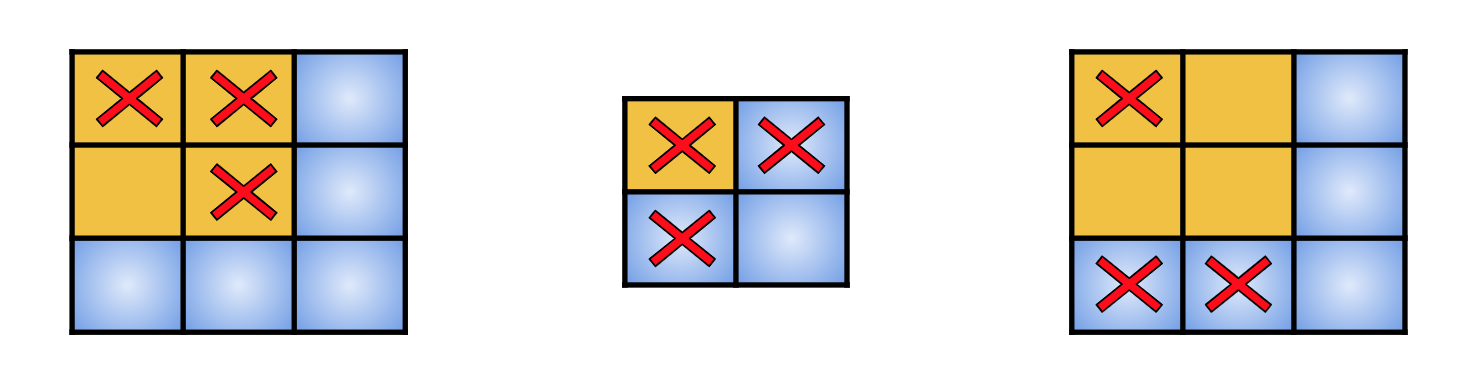}
    \caption{\small Some examples of "interlocking" three straggler configurations. Stragglers can be decoded using a peeling decoder.}
\label{fig:three_stragglers}
\end{figure} 

The following theorem bounds the probability of encountering an undecodable set for local product codes.
\begin{theorem}\label{thm:decodability_prob}
Let $p$ be the probability that a serverless worker straggles independently of others. Let $\bar D$ be the event that a decoding worker working on $n~ (\geq 8)$ blocks in an $(L_A + 1) \times (L_B + 1)$ grid cannot decode. Then, 
\begin{align*}
    \Pr(\bar D) \leq  \sum_{s=4}^{7} \alpha_s p^s (1-p)^{n-s} + \sum_{s=8}^n \binom{n}{s} p^s (1-p)^{n-s},
\end{align*}
where
\begin{align*}
    \alpha_4 &= \binom{L_A + 1}{2}\binom{L_B + 1}{2},~~~~ \alpha_5 = \alpha_4 (n-4),
\end{align*}
{\small
\begin{align*}
\alpha_6 &\leq  \binom{L_A + 1}{3}\binom{L_B + 1}{3}\binom{9}{6} + \alpha_4\binom{n - 4}{2}, \text{and} \\
    \alpha_7 &\leq \binom{L_A + 1}{3}\binom{L_B + 1}{3}\binom{9}{7} + \alpha_4 \binom{n - 4}{3}
\end{align*}
}
\end{theorem}
\begin{proof}
See Section \ref{proof:decodability_prob}.
\end{proof}
In Fig. \ref{fig:undecodability_bound}, the bound in Theorem \ref{thm:decodability_prob} is shown with $p=0.02$ for $L = L_A = L_B = 1, 2, ..., 25$ so that the total number of blocks per worker is $(L + 1)^2$. This shows a "sweet spot" around 121 blocks per decoding worker, or $L=10$, the same choice used in the experiments shown in Fig. \ref{fig:mat_mul_stats}. With this choice of code parameters, the probability of a decoding worker being able to decode all the stragglers is high. This simultaneously enables low encoding and decoding costs, avoids doing too much redundant computation during the multiplication stage (only $21\%$), and gives a high probability of avoiding an undecodable set in the decoding stage. In particular, for $L_A = L_B = 10$, an individual worker is able to decode with probability at least $99.64\%$ when $p=0.02$.

\begin{figure}[t]
    \centering
    \includegraphics[width=.85\linewidth]{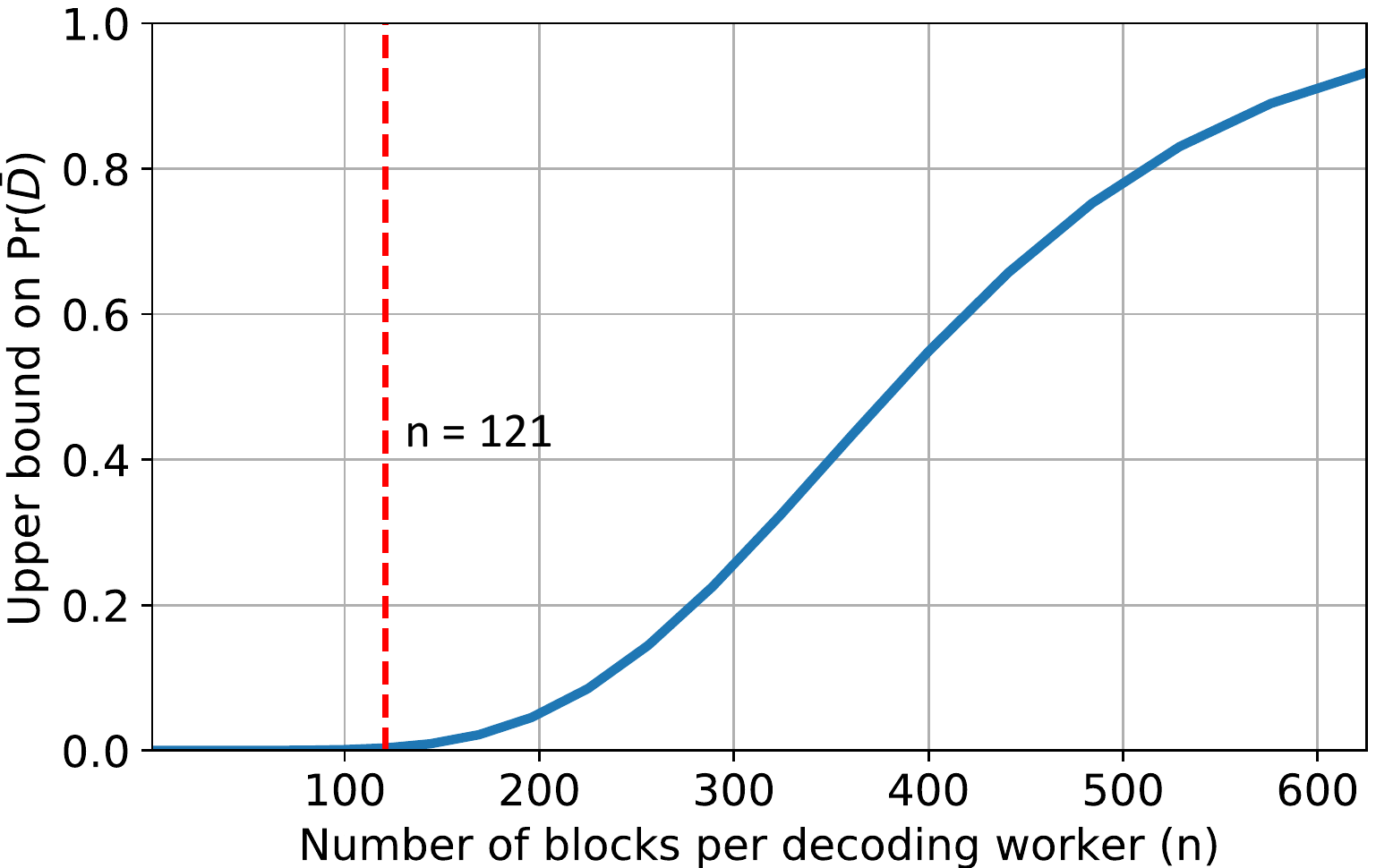}
    \caption{\small Upper bound on probability of the event $\bar D$ (that is, a decoding worker being unable to decode) when $p=.02$. We chose $n=121$ in our experiments which represents a good trade-off between code redundancy and straggler resiliency.}
\label{fig:undecodability_bound}
\end{figure} 

\begin{remark}
The analysis in Sections \ref{sec:decoding_costs} and \ref{sec:straggler_resiliency} derives bounds for one decoding worker. In general, for decoding using $k$ workers in parallel, the respective upper bounds on probabilities in Theorem \ref{thm:decoding_costs} (any decoding worker reading more than $x$ blocks) and Theorem \ref{thm:decodability_prob} (any decoding worker not able to decode) can be multiplied by $k$ using the union bound. 
\end{remark}

\section{Coded Computing in Applications}
In this section, we take several high-level applications from the field of machine learning and high performance computing, and implement them on the serverless platform AWS Lambda. Our experiments clearly demonstrate the advantages of proposed coding schemes over speculative execution.

\subsection{Kernel Ridge Regression}

 \begin{figure*}[t!]
    \centering
    \begin{subfigure}[t]{0.45\textwidth}
        \centering
        \includegraphics[scale=0.45]{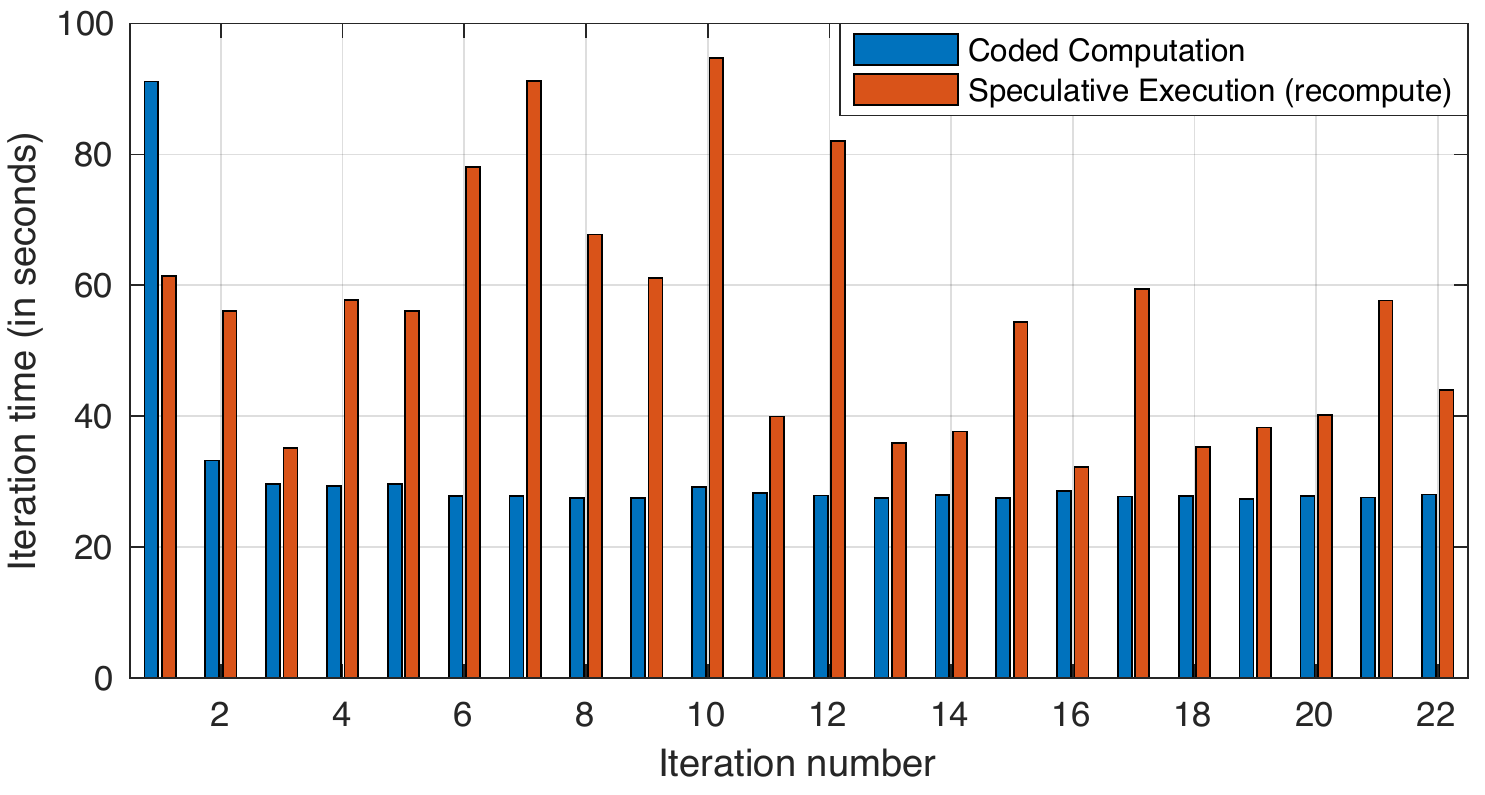}
        \caption{Per iteration time during PCG for ADULT dataset.}
    \end{subfigure}
    ~
    \begin{subfigure}[t]{0.45\textwidth}
        \centering
        \includegraphics[scale=0.48]{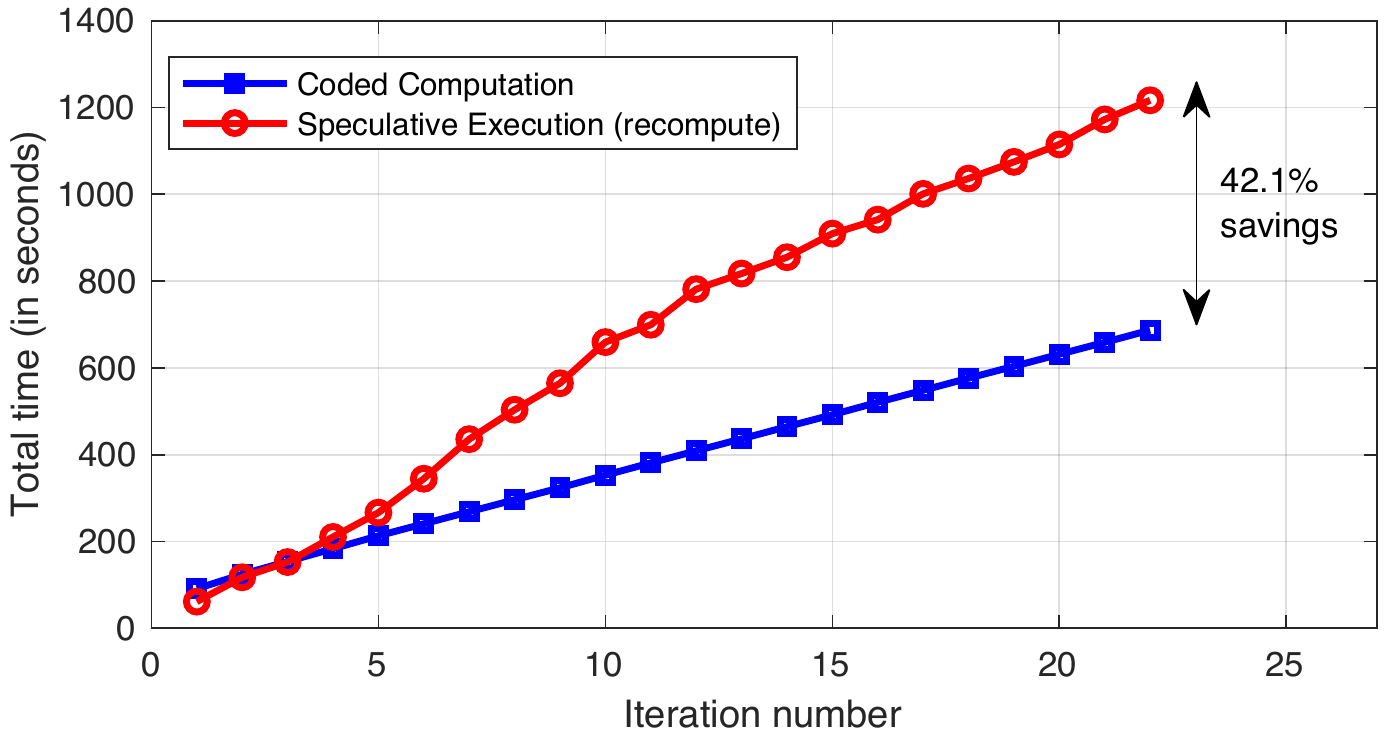}
        \caption{Total running time for PCG for ADULT dataset.}
    \end{subfigure}
    \caption{\small Coded computing versus speculative execution for KRR with PCG on the ADULT dataset. 
    Error on testing dataset was $11\%$.}
    \label{fig:krr_stats_adult}
\end{figure*} 

\begin{figure*}[t!]
    \centering
    \begin{subfigure}[t]{0.45\textwidth}
        \centering
        \includegraphics[scale=0.45]{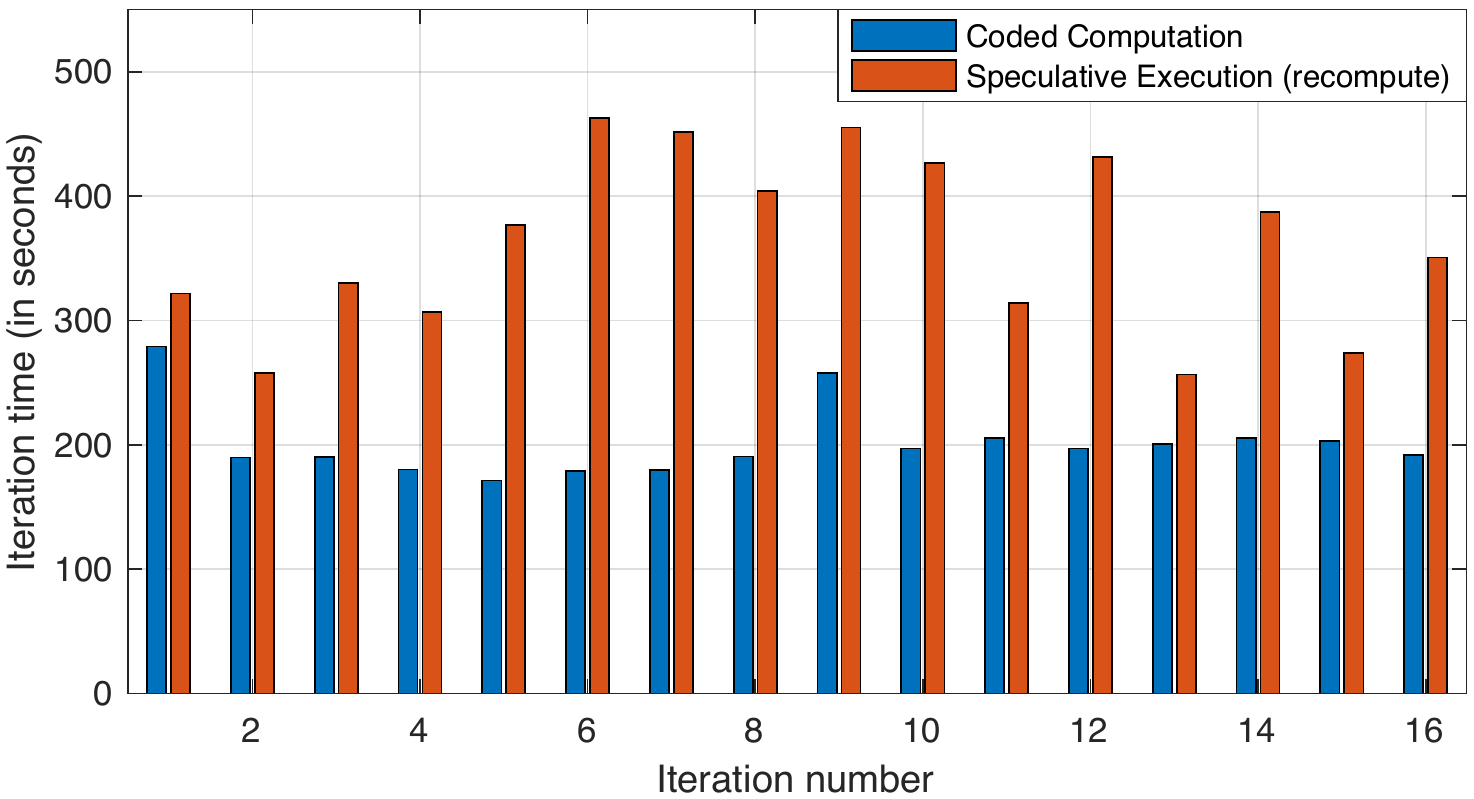}
        \caption{Per iteration time during PCG for EPSILON dataset}
    \end{subfigure}
    ~
    \begin{subfigure}[t]{0.45\textwidth}
        \centering
        \includegraphics[scale=0.47]{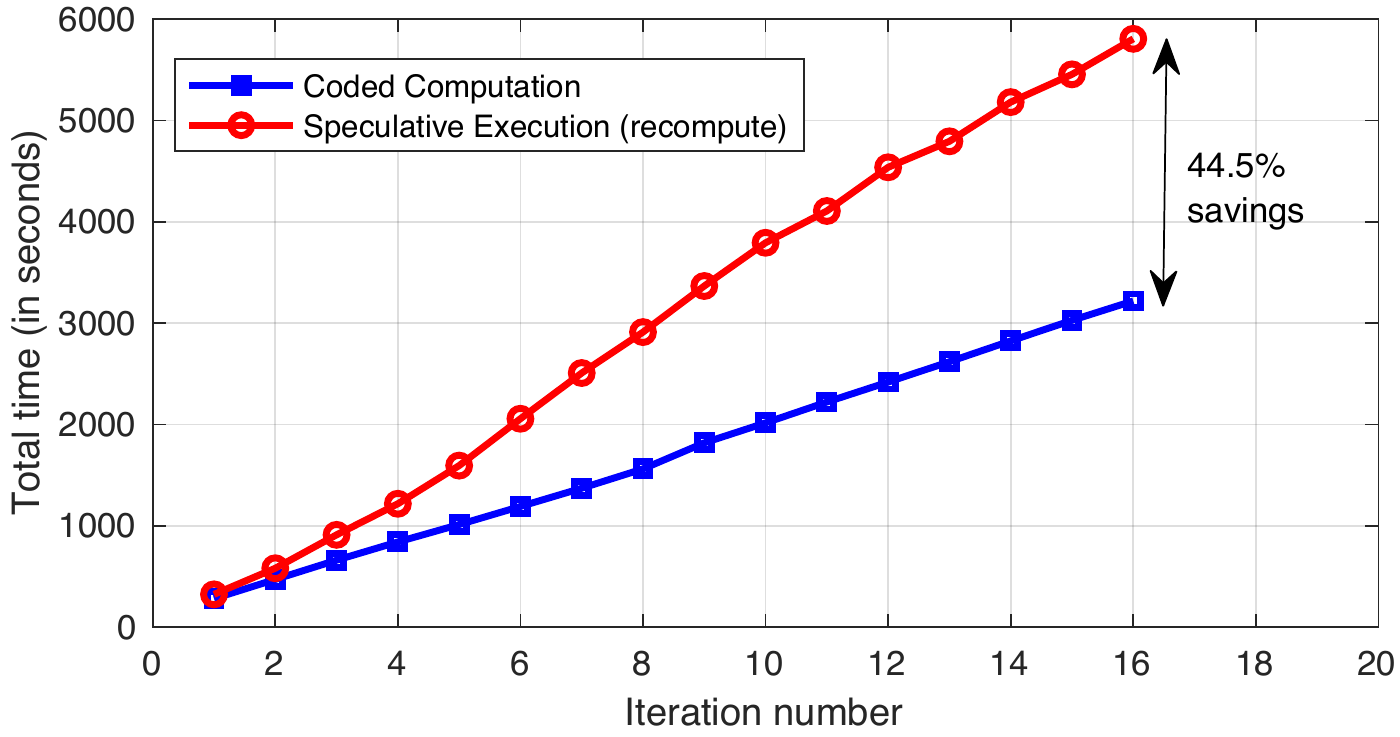}
        \caption{Total running time for PCG for EPSILON dataset}
    \end{subfigure}
    \caption{\small Coded computing versus speculative execution for KRR with PCG on the EPSILON dataset. 
    Error on testing dataset was $8\%$.   }
\label{fig:krr_stats_epsilon}
\end{figure*} 

We first focus on the flexible class of Kernel Ridge Regression (KRR) problems with Preconditioned Conjugate Gradient (PCG). Oftentimes, KRR problems are ill-conditioned, so we use a preconditioner described in \cite{woodruff_pcg} for faster convergence. The problem can be described as
\begin{equation}\label{normal_equation}
(\K + \lambda\I_n)\x = \y,
\end{equation}
where $\K\in \R^{n\times n}$ is a Kernel matrix defined by $\K_{ij} = k(\x_i,\x_j)$ with the kernel function $k: \mathcal{X}\times \mathcal{X}\rightarrow \R$ on the input domain $\mathcal{X} \subseteq \R^d$, $n$ is the number of samples in training data, $\y \in\R^{n\times 1}$ is the labels vector and the solution to coefficient vector $\x$ is desired. A preconditioning matrix $\M$  based on random feature maps \cite{rahimi} can be introduced for faster convergence, so that the KRR problem in Eq. \eqref{normal_equation} can be solved using Algorithm \ref{algo1}. Incorporation of such maps has emerged as a powerful technique for speeding up and scaling kernel-based computations, often requiring fewer than 20 iterations of Algorithm \ref{algo1}   to solve \eqref{normal_equation} with good accuracy.

\begin{algorithm}[t]
\caption{Fast Kernel Ridge Regression using preconditioned conjugate gradient}
 \label{algo1}
\SetAlgoLined
\textbf{Input Data} (stored in S3):
Kernel Matrix $\K\in \R^{n\times n}$ and vector $\y\in \R^{n\times 1}$, regularization parameter $\lambda$,
inverse of the preconditioner $\M\in \R^{n\times n}$ found using the random feature map from \cite{rahimi}\\
\textbf{Initialization}: Define $\x_0 = \mathbf{1}^{n\times1},~~\rr_0 = \y - (\K + \lambda\I_n)\x_0,~~\z_0 = \M^{-1}\rr_0, ~~\p_0 = \z_0$\\
 \While{$\|(\K+\lambda\I_n)\x_k - \y||> 10^{-3}||\y||$}{
 $\h_k = (\K + \lambda\I_n)\p_k$ \tcp*{Computed in parallel using codes}
$\alpha_k = \frac{\rr_k^T\z_k}{\p_k^T\h_k},~~ \x_{k+1} = \x_k + \alpha_k\p_k, ~~\rr_{k+1} = \rr_k - \alpha_k\h_k$\\
$\z_{k+1} = \M^{-1}\rr_{k+1}$ \tcp*{Computed in parallel using codes}
$\beta_k = \frac{\rr_{k+1}^T\z_{k+1}}{\rr_k^T\z_k}$, ~~$\p_{k+1} = \z_{k+1} + \beta_k\p_k$
}
\KwResult{$\x^* = \x_{k+1}$ where $(\K + \lambda\I_n)\x^* = y$ }
\end{algorithm}

\textbf{Straggler mitigation with coding theory}: 
The matrix-vector multiplication in Steps 4 and 6 are the bottleneck in each iteration and are  distributedly executed on AWS Lambda. As such, they are prone to slowdowns due to faults or stragglers, and should be the target for the introduction of coded computation.   To demonstrate the promised gains of the coding theory based approach, we conducted an experiment on the standard classification datasets ADULT and EPSILON \cite{libsvm} with Gaussian kernel $k(\x,\z) = \exp(-||\x- \z||_2^2/2\sigma^2)$ with $\sigma = 8$ and $\lambda = 0.01$, and the Kernel matrices are square of dimension $32,000$ and $400,000$, respectively. We store the training and all subsequently generated data in cloud storage S3 and use Pywren \cite{pywren} as a serverless computing framework on AWS Lambda. 

For this experiment, we implemented a 2D product code similar to that proposed in \cite{tavor} to encode the row-blocks of $(\K + \lambda\I_n)$ and $\M^{-1}$, and distributed them among 64 and 400 Lambda workers, respectively. To compare this coded scheme's performance against speculative execution, we distribute the uncoded row-blocks of $(\K + \lambda\I_n)$ and $\M^{-1}$ among the same number of Lambda workers, and wait for $90\%$ of jobs to finish and restart the rest without terminating unfinished jobs.
Any job that finishes first would submit its results.   The computation times for KRR with PCG on these datasets for the coding-based and speculative execution-based schemes is plotted in  Figs. \ref{fig:krr_stats_adult} and \ref{fig:krr_stats_epsilon}. For coded computation, the first iteration also includes the encoding time. We note that coded computation performs significantly better than speculative execution, with $42.1\%$ and $44.5\%$ reduction in total job times for ADULT and EPSILON datasets, respectively. This experiment again demonstrates that coding-based schemes can significantly improve the efficiency of large-scale distributed computations.
Other regression problems such as ridge regression, lasso, elastic net and support vector machines can be modified to incorporate codes in a similar fashion.

\subsection{Alternating Least Squares}

Alternating Least Squares (ALS) is a widely popular method to find low rank matrices that best fit the given data. This empirically successful approach is commonly employed in applications such as matrix completion and matrix sensing used to build recommender systems \cite{jain2013low}.
For example, it was a major component of the winning entry in the Netflix Challenge where the objective was to predict user ratings from already available datasets \cite{netflix}. 
We implement the ALS algorithm for matrix completion on AWS Lambda using the Pywren framework \cite{pywren}, where the main computational bottleneck is a large matrix-matrix multiplication in each iteration. 

Let $\mathbf R \in \R^{u\times i}$ be a matrix constructed based on the existing (incomplete) ratings, where $u$ and $i$ are the number of users giving ratings and items being rated, respectively. The objective is to find the matrix $\mathbf{\tilde{R}}$ which predicts the missing ratings. One solution is to compute a low-rank factorization based on the existing data, which decomposes the ratings matrix as $\mathbf{\tilde{R}} = \mathbf{H}\mathbf{W},$ where $\mathbf{H} \in \mathbb{R}^{u \times f}, \mathbf{W} \in \mathbb{R}^{f \times i}$ for some number of latent factors $f$, which is a hyperparameter. 

Let us call the matrices $\mathbf{H}$ and $\mathbf{W}$ the \textit{user matrix} and \textit{item matrix}, respectively. Each row of $\mathbf{H}$ and column of $\mathbf{W}$ uses an $f$-dimensional vector of latent factors to describe each user or item, respectively. This gives us a rank-$f$ approximation to $\mathbf{R}$. 
To obtain the user and item matrices, we solve the optimization problem
$\argmin_{\mathbf{H}, \mathbf{W}}  F(\mathbf{H}, \mathbf{W}),$ where the loss $F(\mathbf{H}, \mathbf{W})$ is defined as
\begin{align*}
F(\mathbf{H}, \mathbf{W}) = ||\mathbf{R} - \mathbf{\tilde{R}}||_F^2 + \lambda (||\mathbf{H}||_F^2 + ||\mathbf{W}||_F^2),
\end{align*}
where $\lambda > 0$ is a regularization hyperparameter chosen to avoid overfitting. 
The above problem is non-convex in general. However, it is \textit{bi-convex}\textemdash given a fixed $\mathbf{H}$, it is convex in $\mathbf{W}$, and given a fixed $\mathbf{W}$, it is convex in $\mathbf{H}$. ALS, described in Algorithm \ref{alg:als}, exploits this bi-convexity to solve the problem using coordinate descent. ALS begins with a random initialization of the user and item matrices. It then alternates between a user step, where it optimizes over the user matrix using the current item matrix estimate, and an item step, optimizing over the item matrix using the newly obtained user matrix. Thus, the updates to the user and item matrices in the $k$-th iteration are given by
\begin{align*}
\mathbf{H}_k &= \argmin_\mathbf{H} F(\mathbf{H}, \mathbf{W}_{k-1}) \\
&= \mathbf{R}\mathbf{W}_{k-1}^T(\mathbf{W}_{k-1}\mathbf{W}_{k-1}^T + \lambda \mathbf{I}_f)^{-1}; \\
\mathbf{W}_k &= \argmin_\mathbf{W} F(\mathbf{H}_k, \mathbf{W}) = (\mathbf{H}_{k}^T\mathbf{H}_{k} + \lambda \mathbf{I}_f)^{-1} \mathbf{H}_k^T\mathbf{R}.
\end{align*}

\begin{algorithm}[t]
\SetAlgoLined
 \caption{Alternating Least Squares (ALS)}
 \label{alg:als}
\textbf{Input Data} (stored in S3):
Ratings Matrix $\mathbf{R} \in \R^{u\times i}$, regularization parameter $\lambda$,
latent factor dimension $f$, desired accuracy $\epsilon$\\
\textbf{Initialization}: Define $\mathbf{H}_0 \in \mathbb{R}^{u \times f}$, $\mathbf{W}_0 \in \mathbb{R}^{f \times i}$ with entries drawn independently from a Uniform$[0, 1/f]$ distribution. \\
 \While{$||\mathbf{R} - \mathbf{H}_k\mathbf{W}_k||_F^2 > \epsilon$}{
 User step: $\mathbf{H}_k = \mathbf{R}\mathbf{W}_{k-1}^T(\mathbf{W}_{k-1}\mathbf{W}_{k-1}^T + \lambda \mathbf{I}_f)^{-1}$ \tcp*{Done in parallel using codes} 
Item step: $\mathbf{W}_k = (\mathbf{H}_{k}^T\mathbf{H}_{k} + \lambda \mathbf{I}_f)^{-1} \mathbf{H}_k^T\mathbf{R}$ ~~~~\tcp*{Done in parallel using codes}
}
\KwResult{$\mathbf{H}^* = \mathbf{H}_{k}$, $\mathbf{W}^* = \mathbf{W}_{k}$}
\end{algorithm}
In practice, $u, i \gg f$, so computing and inverting the $f\times f$ matrix in each step can be done locally at the master node. Instead, the matrix multiplications $\mathbf{R}\mathbf{W}_{k-1}^T$ and $\mathbf{R}^T \mathbf{H}_k$ in the user and item steps, respectively, are the bottleneck in each iteration, requiring $\mathcal{O}(uif)$ time. To mitigate stragglers, we use local product codes and speculative execution and compare their runtimes in Fig. \ref{fig:als} for seven iterations. The matrix $\mathbf{R}$ was synthetically generated with $u = i = 102400$ and the number of latent factors used was $f = 20480$. Each rating was generated independently by sampling a Uniform$\{1,2,3,4,5\}$ random variable, intended to be the true user rating. Then, noise generated by sampling a $\mathcal{N}(0, .2)$ distribution was added, and the final rating was obtained by rounding to the nearest integer. 
The ratings matrix $\mathbf{R}$ is encoded once before the computation starts, and thus the encoding cost is amortized over iterations.
We used $500$ workers during the computation phase and $5$ workers during the decoding phase for each matrix multiplication.
It can be seen that codes perform $20\%$ better than speculative execution while providing reliability, that is, each iteration takes on average $\sim$$150$ seconds with much smaller variance in running times per iteration.

\begin{figure}[t!]
    \centering
    \begin{subfigure}[t]{0.45\textwidth}
        \centering
        \includegraphics[scale=0.61]{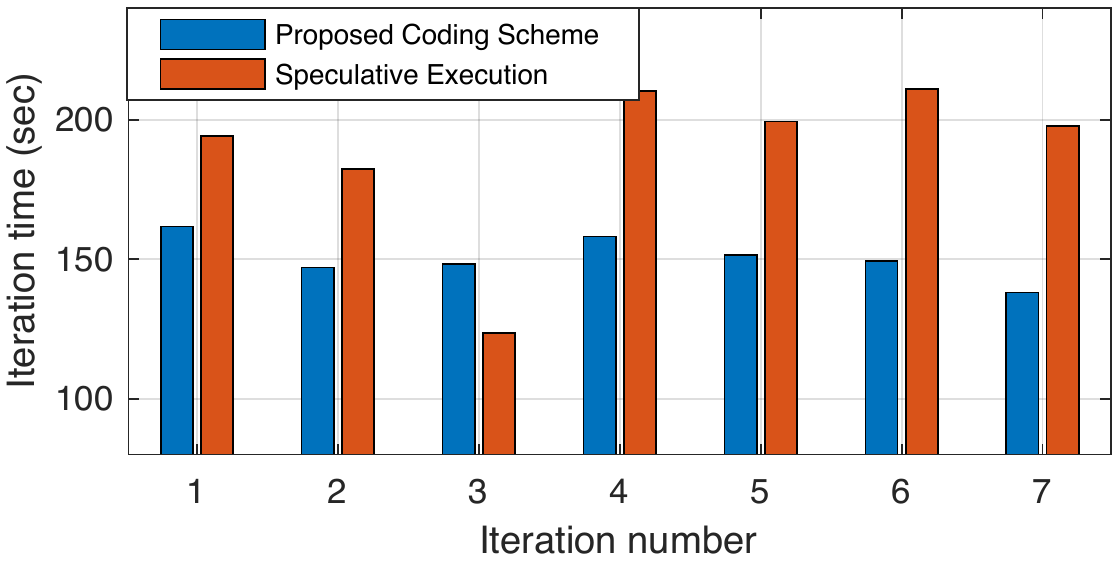}
        \caption{Per iteration time for ALS.}
    \end{subfigure}
    \begin{subfigure}[t]{0.45\textwidth}
        \centering
        \includegraphics[scale=0.6]{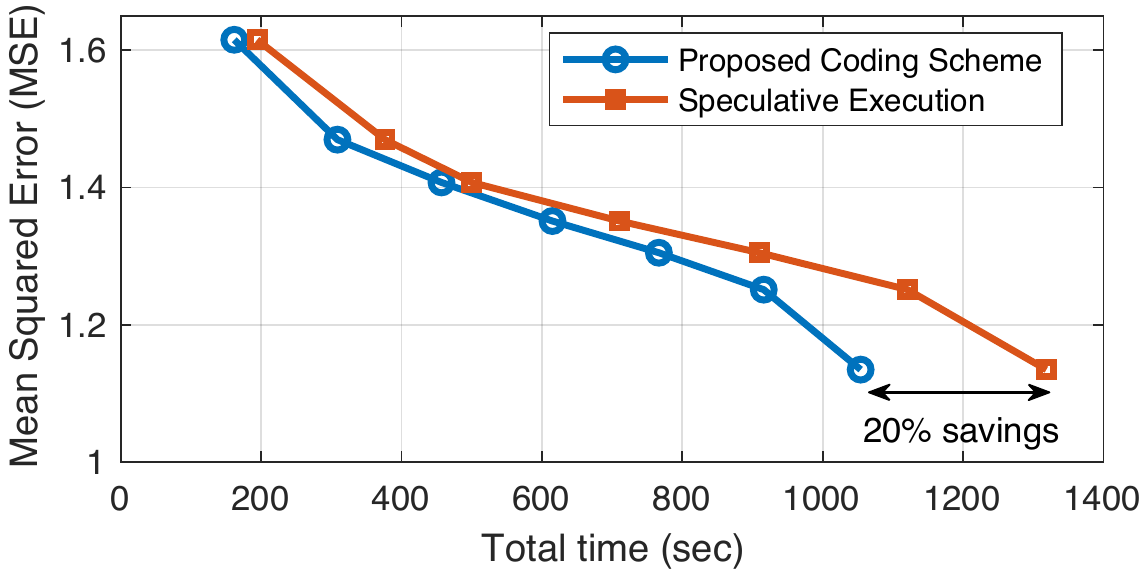}
        \caption{Total running time versus mean squared error for ALS.}
    \end{subfigure}
    \caption{\small Comparison of proposed coding scheme, that is, local product codes, versus speculative execution for straggler mitigation on AWS Lambda. 
    }
    \label{fig:als}
\end{figure} 

\subsection{Tall-Skinny SVD}

Singular Value Decomposition (SVD) is a common numerical linear algebra technique with numerous applications, such as in the fields of  
image processing \cite{svd_ip}, genomic signal processing \cite{svd_genome}, unsupervised learning \cite{pca_svd_ml_app}, and more.  
In this section, we employ our proposed coding scheme in mitigating stragglers while computing the SVD of a tall, skinny matrix $\A\in \R^{m\times p}$, where $m \gg p$. That is, we would like to compute the orthogonal matrices $\mathbf U\in \R^{m\times p}$ and $\mathbf V \in \R^{p\times p}$ and the diagonal matrix $\mathbf\Sigma \in \R^{p\times p} $, where $\A = \mathbf{U\Sigma V}^T$. 

To this end, we first compute the matrix-matrix multiplication $\B = \A^T\A$ which is the main computational bottleneck and requires $\mathcal{O}(mp^2)$ time. Next, we compute the SVD of $\B$. Note that $\B\in \R^{p\times p}$ is a smaller matrix and its SVD $\B = \mathbf{V\Sigma}^2\mathbf{V}^T$ requires only $O(p^3)$ time and memory and can be computed locally at the master node in general. This will give us the matrix $\mathbf V$ and the diagonal matrix $\mathbf \Sigma$. Now, $\mathbf U$ can again be computed in parallel using the matrix-matrix multiplication $\mathbf U = \A\times (\mathbf{V\Sigma}^{-1})$ which requires $\mathcal{O}(mp^2)$ time.\looseness=-1

We compute the SVD of a tall matrix of size $300,000\times 30,000$ on AWS Lambda. For local product codes, we use $400$ systematic workers during computation with $21\%$ redundancy, and $20$ and $4$ workers for parallel encoding and decoding, respectively. For speculative execution, we employed $400$ workers for computing in the first phase and started the second phase (that is, recomputing the straggling nodes) as soon as $79\%$ of the workers from the first phase arrive. Averaged over 5 trials, coded computing took $270.9$ seconds compared to $368.75$ seconds required by speculative execution, thus providing a $26.5\%$ reduction in end-to-end latency.   

Though we do not implement it here, Cholesky decomposition is yet another application that uses matrix-matrix multiplication as an important constituent. It is frequently used in finding a numerical solution of partial differential equations \cite{cholesky_pde},
solving optimization problems using quasi-Newton methods \cite{cholesky_lbfgs}, 
Monte Carlo methods \cite{cholesky_mc}, 
Kalman filtering \cite{cholesky_kf}, etc. 
The main bottleneck in distributed Cholesky decomposition involves a sequence of large-scale outer products \cite{ballardcholesky, numpywren} and hence local product codes can be readily applied to mitigate stragglers.

\section{Proofs}
\subsection{Proof of Theorem \ref{thm:decoding_costs}}
\label{proof:decoding_costs}

To prove Theorem \ref{thm:decoding_costs}, we use a standard Chernoff bound argument. In particular, for any $t > 0$, we can upper bound the probability of reading at least $x$ blocks as

\begin{equation}\label{eq:chernoff}
\Pr(R \geq x) \leq e^{-tx} M_R(t),
\end{equation}

\noindent where  $M_R(t) := \E\left[e^{tR}\right]$ is the Moment Generating Function (MGF) of the random variable $R$. 

We know that the number of blocks read, $R \leq SL$ since we read $\leq L$ blocks every time we decode a straggler. Thus, we can bound $M_R(t)$, the MGF of $R$, in terms of the MGF of $S$, $M_S(\tau) = \E\left[e^{\tau S}\right]$, as
\begin{equation}\label{eq:mgfs}
    M_R(t) = \E\left[e^{tR}\right] \leq \E\left[e^{tL S}\right] = M_S(\tau)|_{\tau = tL}~~\forall~ t>0.
\end{equation}
Since we assume each worker straggles independently with probability $p$, the distribution of $S$ is Binomial$(n, p)$. Thus, its moment generating function is
$M_S(\tau) = \left(1 - p + pe^{\tau}\right)^n.$ Using Eq. \ref{eq:mgfs}, we have
$M_R(t) \leq \left(1 - p + pe^{tL}\right)^n.$
Using this inequality and the fact that $1-y \leq e^{-y}\ ~\forall~ y \in \mathbb{R}$ in the upper bound of Eq. \ref{eq:chernoff}, we get
\begin{equation}\label{eq:upper_bound}
    \Pr(R \geq x) \leq e^{-tx + np - np(\exp(t L))} ~~\forall~ t\geq 0.
\end{equation}
As a last step, we specialize by setting
$t = \dfrac{1}{L} \ln\left(\dfrac{x}{np L}\right),$ which is obtained by optimizing the RHS above with respect to $t$. Substitution into Eq. \ref{eq:upper_bound} gives the desired upper bound on $\Pr(R\geq x)$, proving Theorem \ref{thm:decoding_costs}.

\subsection{Proof of Theorem \ref{thm:decodability_prob}}\label{proof:decodability_prob}

We already discussed in Sec. \ref{sec:straggler_resiliency} that local product codes can decode any three stragglers.
Now, we turn our attention to the case of four or more stragglers. Regardless of how much redundancy is used---including the extreme case of $L_A = L_B = 1$ where every block is duplicated three times---there exist undecodable sets with four stragglers. An example is shown in the middle figure in Fig. \ref{fig:trap_sets}. All 4-undecodable sets come in squares, with every straggler blocking another two off (otherwise, one would be free and decodable, reducing to three stragglers which can always be handled by a peeling decoder). Using this observation, we can create any 4-undecodable set by picking the two rows (from our $L_A + 1$ choices) and two columns (from our $L_B + 1$ choices) to place the stragglers in, yielding exactly four spots. 
Let $\alpha_S$ be the number of undecodable sets with $S$ stragglers. Thus,
$$\alpha_4 = \binom{L_A + 1}{2}\binom{L_B + 1}{2}.$$

All 5-undecodable sets come in the form of 4-undecodable sets with a fifth straggler placed in any vacant spot on the grid. 
This gives us a method to count the number of 5-undecodable sets. First, choose the two rows and two columns that make up the embedded 4-undecodable set. Then, choose from any of the $n - 4$ vacant entries to place the fifth straggler, which gives 
$\alpha_5 = \binom{L_A + 1}{2}\binom{L_B + 1}{2}(n-4)$.

In the case of $S = 6, 7$, undecodable sets can be formed in one of two ways: confining all stragglers to three rows and three columns, or constructing a 4-undecodable set and then placing two (or three for $S=7$) more stragglers anywhere. We can count the former as
\begin{equation}\label{eq:term1}
    \binom{L_A + 1}{3}\binom{L_B + 1}{3}\binom{9}{S}
\end{equation}
for both $S=6$ and $S=7$ since choosing three rows and three columns yields nine blocks, of which we choose $S$. For the latter, we can first construct a 4-undecodable set by picking the two rows and two columns in which to place the stragglers, and then place the remaining $S-4$ anywhere else, giving a total of
\begin{equation}\label{eq:term2}
    \binom{L_A + 1}{2}\binom{L_B + 1}{2}\binom{n - 4}{S - 4}
\end{equation}
such undecodable sets. By summing Eqs. \ref{eq:term1} and \ref{eq:term2}, we obtain an upper bound on $\alpha_S$ for $S=6,7$. This is an upper bound, rather than the exact number of undecodable sets, due to the fact that all sets are counted, but several are overcounted. For example, any 6-undecodable set where all six stragglers are confined to a contiguous $2\times 3$ grid is counted by both terms.

In general, if there are $S$ stragglers, there are $\binom{n}{S}$ ways to arrange the stragglers. Given the number of stragglers $S$, all configurations are equally likely, and the probability of being unable to decode is the percentage of configurations that are undecodable sets. Since $\{\alpha_S\}_{S=4}^{7}$ is the number of $S$-undecodable sets, the probability of being unable to decode given $S (=4,5,6,7)$ stragglers is $\binom{n}{S}^{-1} \alpha_S$.

The probability of encountering eight or more stragglers is small for suitably chosen $L_A, L_B$, owing to the fact that the probability of encountering a straggler is small (for example, $p \approx .02$ for AWS Lambda). Accordingly, we have chosen to focus our analysis on determining $\alpha_S$ for $S \leq 7$. We can obtain an upper bound on the probability of being unable to decode by assuming all configurations where $S \geq 8$ are undecodable sets.
Let $\bar D$ denote the event that a decoding worker cannot decode. Then by the law of total probability,
\begin{align*}
    &\Pr(\bar D) = \sum_{s=0}^{n} \Pr(\bar D | S = s) \Pr(S = s) \\
    &\leq \sum_{s=4}^{7} \binom{n}{s}^{-1} \alpha_s \Pr(S = s) + \sum_{s=8}^{n}\Pr(\bar D | S = s) \Pr(S = s).
\end{align*}
Now using the inequality $\Pr(\bar D | S = s) \leq 1~\forall ~s\geq 8$ and $\Pr(S = s) = \binom{n}{s} p^s (1-p)^{n-s}$ gives the desired upper bound, proving Theorem \ref{thm:decodability_prob}.

\section{Conclusions and Future Work}

In this paper, we argued that in the serverless setting---where communication costs greatly outweigh computation costs---performing some redundant computation based on ideas from coding theory will outperform speculative execution. Moreover, the design of such codes should leverage locality to attain low encoding and decoding costs.
Our proposed  scheme for coded matrix-matrix multiplication outperforms the widely used method of speculative execution and existing popular coded computing schemes in a serverless computing environment. 
All three stages of the coded approach are amenable to a parallel implementation, utilizing the dynamic scaling capabilities of serverless platforms. 
We showed that our proposed scheme is asymptotically optimal in terms of decoding time and further quantified the communication costs during decoding through probabilistic analysis.
Additionally, we derived an upper bound on the probability of being unable to decode stragglers. 

The proposed schemes for fault/straggler mitigation are \textit{universal} in the sense that they can be applied to many existing algorithms without changing their outcome. This is because they mitigate stragglers by working on low-level steps of the algorithm which are often the computational bottleneck, such as matrix-vector or matrix-matrix multiplication, thus not affecting the algorithm from the application or user perspective.
In the future, we plan to devise similar schemes for other matrix operations such as distributed QR decomposition, Gaussian elimination, eigenvalue decomposition, etc.
Eventually, we will create a software library implementing the proposed algorithms for running massive-scale Python code on AWS Lambda. This library would provide a seamless experience for users: they will execute their algorithms on serverless systems (using frameworks such as Pywren \cite{pywren}) as they normally would, and our algorithms can be automatically invoked ``under the hood" to introduce fault/straggler-resilience, 
thus aligning with the overarching goal of serverless systems to reduce management on the user front.


\bibliographystyle{IEEEtran}
\bibliography{bibli}
\end{document}